\documentclass{article}
\usepackage{makeidx,amsfonts}
\newcommand{\ignore}[1]{}
\newcommand{\N}{\mathbb{N}}

\newtheorem{theorem}{Theorem}[section]
\newtheorem{lemma}[theorem]{Lemma}

\newtheorem{corollary}[theorem]{Corollary}
\newtheorem{claim}[theorem]{Claim}

\newenvironment{proof}[1][Proof]{\begin{trivlist}
\item[\hskip \labelsep {\bfseries #1}]}{\end{trivlist}}
\newenvironment{definition}[1][Definition]{\begin{trivlist}
\item[\hskip \labelsep {\bfseries #1}]}{\end{trivlist}}

\begin{document}
\title{Walrasian's Characterization and a Universal Ascending Auction}
\author{Oren Ben-Zwi\footnote{Research was done while the author was in University of Vienna and
was funded by the Vienna Science and Technology Fund (WWTF) through project ICT10-002.
Author's email: \tt{oren.benzwi@gmail.com}}
}

\maketitle

\begin{abstract}
We introduce a novel characterization of all Walrasian price vectors in terms of
forbidden over- and under demanded sets for monotone gross substitute combinatorial auctions.

For ascending and descending auctions we suggest a universal framework for finding
the minimum or maximum Walrasian price vectors for monotone gross substitute
combinatorial auctions. An ascending (descending) auction
is guaranteed to find the minimum (maximum) Walrasian if and only if it follows the
suggested framework.
\end{abstract}

\section{Introduction}

A \emph{combinatorial auction} is a market in which a set of heterogeneous indivisible items is offered
simultaneously to a group of bidders, who may have a different value for every subset of them.
Given a combinatorial auction, we wish to find a \emph{mechanism} that will allocate the items to the bidders
whilst optimizing, for instance, the \emph{social welfare}, the sum of the bidders' evaluations over their allocated bundles.

In a mechanism that we call \emph{iterative auction} the bidders are requested to respond to repeated queries by an
auctioneer until an allocation is determined. The auctioneer queries can come, for instance, in the form of
an item \emph{price vector}; where the bidders respond by a preferable set of items under these prices.
The bidders \emph{utility} for a set of items is the difference between the value of this set and its price.

If in a \emph{price query iterative auction}, the auctioneer is not allowed
to decrease prices, then the auction is called an \emph{ascending auction}.
An ascending auction is a natural and popular auction format in practice\footnote{The
English word ``auction'' is derived from the Latin word \emph{auge{\=o}}
(and its derivations \emph{augere}, \emph{auctum}) which means `increase', `augment'
or `enlarge'.}; whose theoretical properties have been widely studied, see for example~\cite{Cramton1998,CramtonShSt2006,BlumrosenNi2010,BlumrosenNi2005,BlumrosenNi2007,AnderssonAnTa2013,Ausubel2005,Ben-ZwiLaNe2013,KelsoCr1982,GulSt2000}.
In the same manner, a price query iterative auction that cannot increase the prices
is called a \emph{descending auction}\footnote{Ascending and descending auctions,
especially for single items, are also known as \emph{English} and \emph{Dutch} auctions, respectively.}.

Given a price vector, if in an allocation all bidders are allocated a set of maximum utility
the allocation is called an \emph{envy free} allocation and that vector of prices is an \emph{envy free}
price vector. If in addition all items are allocated then the allocation is denoted \emph{Walrasian allocation}
and the price vector is a \emph{Walrasian price vector}~\cite{Walras1874}.
Both envy free and Walrasian naturally `satisfy' the bidders, but the later also maximizes
the global social welfare.  Hence, a Walrasian equilibrium is a most desirable outcome of a combinatorial auction; still, it
is not bound to exist and even if it does, it is not clear how to find it\footnote{In contrast, envy free price vectors
always exist and easy to find, as we can take, for example, all prices to be so high that every bidder will only want an empty set.}.

Arguably, the most intuitive search dynamics for an equilibrium is
by increasing the prices of `excess demand' sets and decreasing the prices
of `excess supply' sets. This process was termed a \emph{t\^{a}tonnement auction}
by Walras~\cite{Walras1874}; however, the notions of excess demand and
excess supply are not easily definable in general, as each bidder may prefer
many different sets for a given price.

By restricting the valuation functions to certain classes, one can guarantee the existence
of Walrasian equilibrium. Lehmann~et~al. presented a hierarchy of bidder valuation classes,
in which the \emph{unit demand} valuation class is a proper subclasses of
the \emph{gross substitute} class~\cite{LehmannLeNi2006}.
\emph{Gross substitute} valuations, introduced by Kelso and Crawford, are valuations
that `support' ascending auctions. Loosely speaking, in a gross substitute valuation, a price
increase cannot decrease the `item demand' of items for which the price was not changed.
A Walrasian equilibrium is guaranteed to exist if all valuations are gross substitute~\cite{KelsoCr1982}.

We will restrict our attention in this work to auctions in which all bidder valuations are monotone gross substitute.
Gul and Stacchetti proved that these valuations are maximal to guarantee equilibrium, by
showing that for any non gross substitute valuation, there exists a set of unit demand valuations, such that together
all valuations cannot guarantee equilibrium~\cite{GulSt1999}. They also showed that all Walrasian
price vectors form a `lattice' and hence have unique \emph{minimum} and \emph{maximum} vectors.
In another work, Gul and Stacchetti defined \emph{over demanded sets} for monotone gross substitute valuations and showed that a price vector is
\emph{envy free} if and only if there is no over demanded set~\cite{GulSt2000}. Lastly, they introduced a unique \emph{excess demand set}
and showed that an `excess demand ascending auction', an auction that updates only prices of the excess demand set, will result in the minimum Walrasian.

We use Gul and Stacchetti's \emph{over demanded sets}, together with a dual \emph{under demanded sets} definition and show that a price vector is
\emph{Walrasian} if and only if there is no over demanded or under demanded set. We show further that a price vector is also \emph{minimum} (\emph{maximum})
Walrasian if and only if there is no \emph{weakly} under (over) demanded sets.

In other words, we generalize Gul and Stacchetti's price characterization from \emph{envy free} price
if and only if there is no \emph{over demanded} set, to \emph{Walrasian} price vector
if and only if there is no \emph{over demanded} nor \emph{under demanded} set.

We also introduce a generalize \emph{excess demand} set system (where Gul and Stacchetti's excess demand is an element in this system) and prove
that an ascending auction will find the minimum Walrasian price vector if and only if it will increase prices of excess demand sets. We show the same
for descending auctions.

\section{Preliminaries}
We introduce some of the definitions we use throughout the sequel. For a
complete view we add more definitions and known results on Appendix~\ref{app::more}.
\subsection{Combinatorial auctions (the model)}
In a \emph{combinatorial auction} setting we are given $m$ items, $\Omega$,
and $n$ bidders, where each bidder $i\in [n]$ is associated with a valuation functions\footnote{It
is natural to assume that valuations are non-negative rationals
and then, maybe after some scaling, one can assume even that valuations are natural numbers (including $0$). We will assume later
also that the prices are natural numbers, but this assumption is not needed given the first one. Our results will hold
also for real number valuations, but the proofs need to be more delicate.} $v_i:2^{\Omega}\rightarrow \N$.
The valuation functions are normalized such that
$\forall i,\ v_i(\emptyset) = 0$; and monotone with respect to inclusion. 

Bidder utilities are defined using an item \emph{price vector} $p = (p_1,p_2,...,p_m)\in \N^m$.
Given $p\in \N^m$, the \emph{utility} of bidder $i$ for a set of items $S\subseteq \Omega$
is defined to be $u_{i,p}(S) = v_i(S) - p(S)$, where $p(S) = \sum_{j\in S}{p_j}$.
A set $S\subseteq \Omega$ is \emph{demanded} by bidder $i$ at price $p$ if it maximizes the utility,
i.e., $u_{i,p}(S) \geq u_{i,p}(T), \forall T\subseteq \Omega$.
The set of all demanded sets for bidder $i$ at $p\in \N^m$ is denoted
$D_i(p) = \{D| u_{i,p}(D) \geq u_{i,p}(T), \forall T\subseteq \Omega\}$.

The bidders' incentive comes from the \emph{maximum utility} $u_i(p)$ function, which
is the utility of any demanded set $D$, i.e., $u_i(p) = u_{i,p}(D)$.
We omit the bidder's subscript and note $u_i(p)$ as $u(p)$ and $D_i(p)$ as $D(p)$ when
the bidder is given or plays no role. For singletons, we omit curly brackets and
use $j$ instead of $\{j\}$.

A \emph{solution} or an \emph{outcome} of a combinatorial auction is an allocation of
the items to the bidders.
The \emph{social welfare} of a solution is the sum of bidders valuations over
their allocated sets; that is, for an allocation $\{S_i\}_{i\in [n]}$ the
social welfare is $\sum_i{v_i(S_i)}$. $\ $

\subsection{Unit demand valuation}
A \emph{unit demand} valuation $v$ is defined by $m$ values $v(j),\ \forall j\in \Omega$;
with $v(S) = \max_{j\in S}{v(j)}$, for $S\subseteq \Omega$.

Assume all valuations are unit demand; we use the following two functions
for $S\subseteq \Omega$ and $p\in \N^m$: $\Lambda_p(S) = \{i | D_i(p) \subseteq S\}$ and
$\Xi_p(S) = \{i | D_i(p) \cap S \neq \emptyset\}$.
Mishra and Talman define $S\subseteq \Omega$ to be over demanded for $p\in \N^m$ if
$|\Lambda_p(S)| \geq |S|$ and under demanded if $|\Xi_p(S)| \leq |S|$\footnote{Note the resemblance
to Hall's marriage theorem~\cite{Hall1935}.}~\cite{MishraTa2010}.
Andersson~et.~al define $S\subseteq \Omega$ to be in excess demand for $p$ if
$\forall T\subseteq S, |\Lambda_p(S)\cap \Xi_p(T)| > |T|$~\cite{AnderssonAnTa2013}.

\subsection{Walrasian equilibrium and gross substitute valuation}
A $p\in \N^m$ and a \emph{partition}\footnote{Recall that a partition is an exact allocation, that is, an
allocation that allocates all items. In general, a Walrasian allocation does not have to be a partition, but for
monotone valuations it can be assumed so.} $\{S_i\subseteq \Omega\}_{i\in [n]}$ such that every bidder is allocated
a demand set in $p$ is called \emph{Walrasian equilibrium}.
Walrasian equilibrium does not have to exist but if it does it is known that it also
maximizes the social welfare (a variant of the \emph{first welfare theorem}, see for example~\cite{BlumrosenNi2007}).
Maximizing the social welfare can be addressed by the following integer linear program, $IP$:

\[\max
\sum_{i \in [n],\, S \subseteq \Omega}
x_{i,S} \cdot v_i(S)\]
\[
\mbox{ s.t. } \quad \quad \enspace \sum_{i \in [n],\, S \mid j \in S}{x_{i, S}} \le 1 \quad \forall j \in \Omega
\]
\[\quad \quad \quad \quad \quad \quad \sum_{S \subseteq \Omega} x_{i, S} \le 1 
\quad \quad \forall i \in [n]
\]
$$\quad \quad \quad \quad \quad \quad \quad \quad \quad x_{i,S} \in \{0,1\} \quad \quad \quad \forall i \in [n],\, S\subseteq \Omega$$
%
%
%
If we relax integrality constraints of the program we get a linear program whose
dual, $DLP$, is:

\[\min
\sum_{i \in [n]} \pi_i + \sum_{j \in \Omega} p_j \quad \quad \quad
\]
\[s.t. \quad \quad \quad
\pi_i + \sum_{j \in S} p_j \ge v_i(S)\quad
\forall i \in [n],\, S \subseteq \Omega
\]
$$\quad \quad \quad \quad \pi_i \ge 0,\, p_j \ge 0 \quad \quad \quad \forall i \in [n],\, j \in \Omega$$

Bikhchandani and Mamer~\cite{BikhchandaniMa1997} observed that a
Walrasian equilibrium exists if and only if the value of the maximum social
welfare, i.e., a solution to the $IP$, equals the optimum in the dual problem $DLP$. Moreover, in this case,
the set of the optimal dual variables $\{p_j\}_{j \in \Omega}$ is a
Walrasian price vector and the variables $\pi_i$ are such that $\pi_i = u_i(p)$
the maximum utility functions for $p$, hence determined by $p$.

For $p\in \N^m$ we mark a lower bound on
the dual objective (the \emph{Lyapunov}) by $L(p) = \sum_i{u_i(p)} + \sum_j{p_j}$.
This lower bound determines feasible dual variables for every $p$
and is tight for an optimum\footnote{Note that the Lyapunov is dual feasible
by definition, hence it is also an upper bound for the social welfare value.}.
%
The following is then based on Bikhchandani and Mamer~\cite{BikhchandaniMa1997}
and Ausubel~\cite{Ausubel2005}.
\begin{corollary}\label{cor::L::min::is::Wal}
$p\in \N^m$ is Walrasian if and only if its Lyapunov
$L(p)$ is minimum with respect to all possible prices $p$.
\end{corollary}

We look at the space of (vectors from or) functions on $\Omega$ as ordered by the
domination order, namely, for ${p,q}: \Omega \mapsto \N$, $p \leq q$
if for every $j \in \Omega$, $p_j \leq q_j$. If $p \leq q$ we also
say that $p$ is dominated by $q$.
For two price vectors $p,q\in \N^m$ we note by $p\wedge q$ and $p\vee q$ the minimum (meet) and
the maximum (join) coordinate-wise, respectively. That is, $(p\wedge q)_j = \min\{p_j,q_j\}$ and
$(p\vee q)_j = \max\{p_j,q_j\}$.

For two price vectors $p,q\in \N^m$ and an $S\subseteq \Omega$, let $S^=(p,q) = \{j\in S | p_j = q_j\}$.
The \emph{gross substitute} class of valuation is the class in which a player never
drops an item whose price was not increased in an ascending auction dynamics; formally:
\begin{definition} [Gross substitute~\cite{KelsoCr1982}]\label{def:gs}
A valuation is \emph{gross substitute} if for every $p\in \N^m$ and  $S\in D(p)$, for every price vector $q \geq p$, $\exists S'\in
D(q)$ such that $S^=(p,q)\subseteq S'$. 
\end{definition}

In other words, if $S$ is a demanded set for $p$, and $q \geq p$ then there is a demand set for $q$
that contains all elements $j \in S$ for which the price did not change. 
Kelso and Crawford proved that if all valuations are gross substitute,
then a Walrasian equilibrium is guaranteed to exist~\cite{KelsoCr1982}.

\subsection{Ascending and descending auctions}
An \emph{ascending auction} is an iterative algorithm which starts with a low
price vector (we assume $p = 0^m$, the $0$ vector) and in each iteration
finds a set of items to increase, we assume the increase is one unit
of price for every item in the increased set. We denote the increase of
each item in $S\subseteq \Omega$ by one unit from $p\in \N^m$ by $p+1_S$, that is, for each $j\in \Omega$,
\[ (p+1_S)_j = \left\{\begin{array}{ll}
      p_j + 1 & \mbox{ if $j\in S$}\\
      p_j & \mbox{ else}\\
   \end{array} \right. \]
In the same manner, a \emph{descending auction} is an iterative algorithm which starts with a high
price vector, we can assume $\forall j\in \Omega, p_j = \max_i{v_i(\Omega)}$. In each iteration the
algorithm finds a set of items to decrease, where the decrease is again one unit
of price for every item in the decreased set. We denote the decrease of
each item in $S\subseteq \Omega$ by one unit from $p\in \N^m$ by $p-1_S$ where,\footnote{We implicitly assume
that if $p\in \N^m$ is decreased for $S\subseteq \Omega$, then $\forall j\in S, p_j\geq 1$.}
\[ (p-1_S)_j = \left\{\begin{array}{ll} 
      p_j - 1 & \textrm{ if $j\in S$}\\
      p_j & \textrm{ else}\\
   \end{array} \right. \]

\section{Accommodating the Lyapunov}\label{chap::char}
\subsection{Over- and under-demand}
Next we define two important functions that will help us characterize
the Walrasian prices and later also characterize a new framework for ascending and
descending auctions.
Gul and Stacchetti defined a \emph{requirement} function that we will use in order
to define \emph{over-demanded sets}; and then we will define a dual \emph{redundant}
function and use it to define \emph{under-demanded sets}.
For intuition, we think about the demand of a bidder when we increase
the prices of a set of items. Clearly, we know that $\forall p\in \N^m, S,T\subseteq \Omega, u(p+1_S) \geq u_p(T) - |S\cap T|$,
as the left hand side is the maximum utility over all possible sets and the right hand side is the utility of a specific set.
Specifically, taking $T\in D(p)$, a new demand set can arise but we can
be assured that the worse case for that bidder is just an old demand
set which intersect the least with the set of items whose prices were increased.
Same thing with decreasing prices, $\forall p\in \N^m, S,T\subseteq \Omega, u(p-1_S) \geq u_p(T) + |S\cap T|$,
here the worse is just a demand set
that intersect the most with the set of items whose prices were decreased.
We will see later that in terms of cardinality, for monotone gross substitute valuations,
these measures are accurate.

\begin{definition} [Requirement function~\cite{GulSt2000}]
Given $p\in \N^m$ and bidder $i$, the requirement of $i$ on $p$
for $S\subseteq \Omega$ is defined to be $l_{i,p}(S) = \min_{D\in D_i(p)}{|S\cap D|}$.
The \emph{auction requirement} on $p$ for $S\subseteq \Omega$ is defined to be
$l^p(S) = \sum_{i\in [n]}{l_{i,p}(S)}$.
\end{definition}

\begin{definition} [Redundant function]
For $p\in \N^m$ and bidder $i$, the redundant of $i$ on $p$
for $S\subseteq \Omega$ is defined to be $h_{i,p}(S) = \max_{D\in D_i(p)}{|S\cap D|}$.
The \emph{auction redundant} on $p$ for $S\subseteq \Omega$ is defined to be
$h^p(S) = \sum_{i\in [n]}{h_{i,p}(S)}$.
\end{definition}

When the bidder is known we use $l_p(\cdot)$ and $h_p(\cdot)$ (sub-script $p$) as the requirement and redundant
of a bidder on $p$ for a set of items; we use $l^p(\cdot)$ and $h^p(\cdot)$ (supper-script $p$)
only for the auction requirement and auction redundant functions (of \emph{all} bidders).
With the help of the requirement and redundant functions we are ready to define
\emph{over-} and \emph{under-}demanded sets.
\begin{definition} [Over-demand~\cite{GulSt2000}]
Given $p\in \N^m$ the set $S\subseteq \Omega$ is \textup{(}weakly\textup{)} over-demanded if \textup{(}$l^p(S)\geq |S|$\textup{)}, $l^p(S) > |S|$.
\end{definition}

\begin{definition} [Under-demand]
Given $p\in \N^m$ the set $S\subseteq \Omega$ is \textup{(}weakly\textup{)} under-demanded if \textup{(}$h^p(S)\leq |S|$\textup{)}, $h^p(S) < |S|$.
\end{definition}
For $p\in \N^m$ we note the collection of all over-demanded, (weakly over-demanded)
and under-demanded, (weakly under-demanded) sets by $OD(p)$, ($WOD(p)$) and $UD(p)$, ($WUD(p)$), respectively.

\subsection{New gross substitute definitions}
More than three decades since the class of gross substitute was introduced
by Kelso and Crawford~\cite{KelsoCr1982} we gradually reveal its enormous
importance and, at the same time, the tremendous effort we still need to overgo in order to clearly
understand it. After some seminal initial works, such as Gul and Stacchetti~\cite{GulSt1999,GulSt2000},
another keen step towards this end was made recently by Paes-Leme
in a survey that sketches abundant interesting class' properties and connects plentiful different
intriguing definitions~\cite{PaesLeme2014}. We present yet two more novel definitions
and hope that their introduction will serve as another step towards this goal as well.

Consider again what happens to a single bidder's utility function on $p\in \N^m$ if we increase or decrease $p$
for all items of a predefined $S\subseteq \Omega$. It is clear by the definitions of $l_p(S)$ and $h_p(S)$ that
$u(p - 1_S) \geq u(p) + h_p(S)$ and $u(p + 1_S) \geq u(p) - l_p(S)$, since the sets that achieved the maximum
or minimum for $p$ are still available for $p - 1_S$ and $p + 1_S$.
We show next that monotone gross substitute valuations are the \emph{only} valuations for which these two inequalities are equalities
for every $p\in \N^m$ and $S\subseteq \Omega$.

\begin{theorem}\label{GS::def}
Given a valuation over $2^{\Omega}$ the following are equivalent.
\begin{enumerate}
	\item	$\forall p\in \N^m$ and $S\subseteq \Omega$, $u(p) = u(p - 1_S) - h_p(S)$
	\item	$\forall p\in \N^m$ and $S\subseteq \Omega$, $u(p) = u(p + 1_S) + l_p(S)$
	\item	The valuation is monotone gross substitute.
\end{enumerate}
\end{theorem}
\begin{proof}
We will see a cyclic derivation
\begin{itemize}
	\item $1\Rightarrow 2$\\
	By definition, $\forall p\in \N^m$ and $S\subseteq \Omega$, $u(p) \leq u(p + 1_S) + l_p(S)$.
	Assume $\exists p\in \N^m$ and $S\subseteq \Omega$, $u(p) < u(p + 1_S) + l_p(S)$ and look now on
	the price vector $p+1_S$ and a set $D\in D(p+1_S)$ such that $D = \arg\max_{D\in D(p+1_S)}{|D\cap S|}$.
	Since $u(p + 1_S) = u(p + 1_S - 1_S) - h_{p+1_S}(S)$ we have that $l_p(S) > h_{p+1_S}(S)$. But by definition
	of $D$, $u(p) = u(p + 1_S) + |S\cap D| = u_p(D)$, hence, $D\in D(p)$ and clearly, $l_p(S) \leq h_{p+1_S}(S)$.
	\item $2\Rightarrow 3$\\
	Again, by definition of $l_p(S)$ it holds that $u(p) \leq u(p+1_S) + l_p(S)$. We show that for a non-gross
	substitute valuation there exist $p\in \N^m$ and an $S\subseteq \Omega$ such that $u(p) < u(p+1_S) + l_p(S)$.
	Let $q,A,j_1,j_2,j_3$ be as described in Theorem~\ref{non::gs} and recall that either $D(q) = \{A,A\cup\{j_1,j_2\}\}$
	or $D(q) = \{A\cup\{j_3\}, A\cup\{j_1,j_2\}\}$. By monotonicity $v(A) \leq v(A\cup \{j_1\})$
	(respectively, $v(A\cup \{j_3\}) \leq v(A\cup \{j_1,j_3\})$) hence the fact that $A\cup \{j_1\} \notin D(q)$
	(resp., $A\cup \{j_1,j_3\} \notin D(q)$) means that $q_{j_1} > 0$. For $p = q - 1_{j_1}$ and $S = \{j_1,j_2\}$
	we have that $\forall D\in D(p),\ j_1,j_2\in D$; or else $u(q) < u_p(A\cup\{j_1,j_2\}) \leq u_p(D) \leq u_q(D) + 1 \leq u(q)$,
	where the last inequality comes from the assumption that $j_1,j_2\in D$ does not hold and the fact that $u_p(A) < u_p(A \cup\{j_1,j_2\})$
	and $u_p(A\cup \{j_3\}) < u_p(A \cup\{j_1,j_2\})$, so $D\notin D(q)$. Hence we conclude that $l_p(S) = 2$ and therefore,
	$u(p) = u_q(A \cup\{j_1,j_2\}) + 1 < u_q(A \cup\{j_1,j_2\}) + 2 \leq u(p+1_S) + l_p(S)$.
	\item $3\Rightarrow 1$\\
	Indeed, by definition of the redundant function $h$, $u(p) \leq u(p - 1_S) - h_p(S)$.
	We will show next that for any $S'\subseteq \Omega$, $u(p) \geq u_{p - 1_S}(S') - h_p(S)$
	hence the direction will follow. Assume $u(p) - u_p(S') = k$, we show that $\exists R,D\subseteq \Omega$
	such that $D\in D(p),\ |R|\leq k$ and $S'\subseteq D\cup R$. Note that then $|S'\cap S|\leq h_p(S) + k$
	so $u_{p - 1_S}(S') \leq u_p(S') + h_p(S) + k \leq u(p) + h_p(S)$ as claimed. To show the existence of
	$D,R$ we use induction on $k$. For $k = 0$ just take $D = S'$ and $R = \emptyset$. So $k > 0$ hence by
	the single improvement property $\exists T$ such that $u_p(T) > u_p(S')$ and $|S'\setminus T| \leq 1$.
	Now since $u_p(T) > u_p(S')$ we can apply the induction hypothesis on $T$ and find two sets $D,R'$ such
	that $D\in D(p)$, $|R'| \leq u(p) - u_p(T)$ and $T\subseteq D\cup R'$. We claim that the same set $D$
	and the set $R = R'\cup (S'\setminus T)$ suffice. Indeed, $S'\subseteq D\cup R$, $D\in D(p)$ and
	$|R|\leq u(p) - u_p(T) + 1\leq u_p - u_p(S')\leq k$.
	Hence, for monotone gross substitute valuation it holds that $u(p) = u(p - 1_S) - h_p(S)$.
\end{itemize}
\end{proof}

The next two corollaries play a key role in the following \emph{price characterization} result; the first one
follows the definitions of the Lyapunov and the redundant/requirement of \emph{all} bidders; adding the
over-/under-demand and Corollary~\ref{cor::L::min::is::Wal} we derive the second.
\begin{corollary} [of Theorem~\ref{GS::def}]\label{L::char} 
Assume all valuations are monotone gross substitute, then for each $p\in \N^m$ and each $S\subseteq \Omega$
\begin{enumerate}
	\item	$L(p-1_S) = L(p) + h^p(S) - |S|$
	\item	$L(p+1_S) = L(p) - l^p(S) + |S|$
\end{enumerate}
\end{corollary}

\begin{corollary} [of Corollary~\ref{cor::L::min::is::Wal} and Corollary~\ref{L::char}]\label{cor::cor}
When all valuations are monotone gross substitute, if for $p\in \N^m$ there exists an under- or over-demanded set, then $L(p)$
is not minimum and $p$ is not a Walrasian price vector.
\end{corollary}

\section{Price characterization}\label{sec::price}
We present now the main Walrasian price characterization. This theorem in turn, is the main
building block in the correctness' proof of our ascending and descending auctions framework later.
\begin{theorem}\label{price::char}
If all valuations are monotone gross substitute then given $p\in \N^m$ the following hold:
\begin{itemize}
	\item	$p$ is Walrasian if and only if there is no over-demanded and no under-demanded set.
	\item	$p$ is minimum Walrasian if and only if there is no over-demanded and no non-trivial weakly under-demanded set.
	\item	$p$ is maximum Walrasian if and only if there is no non-trivial weakly over-demanded and no under-demanded set.
\end{itemize}
\end{theorem}
\begin{proof}
One direction of the first item is Corollary~\ref{cor::cor}. Indeed, if $p\in \N^m$ is a
Walrasian price vector then by Corollary~\ref{cor::L::min::is::Wal} the dual objective $L(p)$ is minimum and hence there cannot be any under- or over-demanded
set.
For the other direction assume that for $p\in \N^m$ there are no over-demanded or under-demanded sets;
we will show that $L(p)$ is minimum, hence, by Corollary~\ref{cor::L::min::is::Wal}, $p$ is Walrasian.
Assume $L(p)$ is not minimum and let $q\in \N^m$ be the `closest' such that $L(q) < L(p)$, that is,
$q$ minimizes (over all price vectors $\widehat{q}\in \N^m$, such that $L(\widehat{q})<L(p)$) the $dist(p,q) = \sum_j{|p_j - q_j|}$.
Let $x^-$ be a vector such that $x^-_j = \max{\{p_j - q_j,0\}}$ and let $x^+$ be a vector such that $x^+_j = \max{\{q_j - p_j,0\}}$.
Let $p^- =  p - x^-$ and $p^+ = p + x^+$.
\begin{claim}
Either $p$ dominates $q$ ($p\geq q$) or $q$ dominates $p$ ($q\geq p$).
\end{claim}
\begin{proof}
Assume not, that is, $p^- \neq p$ and $p^+ \neq p$; and by definition of $q$, $L(p^-) = L(p^+) = L(p)$.
We will use the submodularity of the Lyapunov, Lemma~\ref{L::SM}, in order to show that $L(q) \geq L(p)$ which is a contradiction
to $q$'s definition, hence, either $x^- = 0^m$ and $p\leq q$ or $x^+ = 0^m$ and $p\geq q$.

By the Lyapunov submodularity $L(p \wedge q) + L(p \vee q) \leq L(p) + L(q)$. But since $p \wedge q = p^-$
and $p \vee q = p^+$ and the fact that $L(p^-) = L(p^+) = L(p)$, we come to a contradiction.
\end{proof}

Let $R = \{j| p_j \neq q_j\}$ and note that if $q = p + 1_R$ or $q = p - 1_R$ then by Corollary~\ref{L::char} and the fact that there are
no over-demanded or under-demanded sets $L(p)\leq L(q)$ contradicting the existence of such a $q$.
Assume now that $q\geq p$. Let $R_1 = R$, $R_2 = \{j| q_j - p_j \geq 2\}$, $R_3 = \{j| q_j - p_j \geq 3\}\ \cdots\ R_{\tau} = \{j| q_j - p_j \geq \tau\}$.
Note that $R_1 \supseteq R_2 \supseteq R_3 \supseteq \cdots \supseteq R_{\tau}$ and that $q = p + 1_{R_1} + 1_{R_2} + \cdots + 1_{R_{\tau}}$.
By repeatedly applying Corollary~\ref{L::char} we get that $L(q) = L(p) - \sum_{k\in [\tau]}{l^{p + \sum_{\eta < \tau}{1_{R_\eta}}}(R_k)} + \sum_{k}{|R_k|}$,
hence, $\sum_{k\in [\tau]}{l^{p + \sum_{\eta < \tau}{1_{R_\eta}}}(R_k)} > \sum_{k}{|R_k|}$.
By a corollary of Theorem~2 of Gul and Stacchetti~(see Theorem~\ref{thm::2::GS} in the appendix) it can be shown that $l^{p + \sum_{\eta < \tau}{1_{R_\eta}}}(R_k) \leq l^p(R_k)$.
Together we conclude that $\sum_{k\in [\tau]}{l^p(R_k)} > \sum_{k}{|R_k|}$, so $\exists \eta$ for which $l^p(R_{\eta}) > |R_{\eta}|$ which contradicts
the fact that $q$ is the closest to $p$. The case for $p \geq q$ is symmetric, therefore omitted.

Next we show that there is also no non-trivial weakly under-demanded set if and only if $p$ is minimum Walrasian.
Assume first that there is no non-trivial set of weakly under-demand and that exists $q\in \N^m$ such that $q \not\geq p$
and $L(q) = L(p)$.
Recall that there is no non-trivial weakly under-demanded set means that $\forall S\neq \emptyset,\ h^p(S) > |S|$
and that the lattice property of Walrasian prices (Lemma~\ref{lattice}) means that if $\exists q \not\geq p$
such that $L(q) = L(p)$, then $\exists q' \leq p$ such that $q'\neq p$ and $L(q') = L(p)$.
By convexity and the fact that $L(p)$ is minimum, we can assume that $q' = p - 1_S$, but then by Corollary~\ref{L::char},
$S$ is weakly under-demanded, a contradiction.
For the other direction, if $p$ is minimum Walrasian then $\forall q \leq p,\mbox{ if } p\neq q \mbox{ then } L(p) < L(q)$.
In particular, $\forall S\neq \emptyset,\ L(p) < L(p - 1_S)$; so, again by Corollary~\ref{L::char},
$S$ is not weakly under-demanded. This completes the proof that there is no non-trivial weakly under-demanded set if and only if $p$ is minimum Walrasian.
The proof that maximum Walrasian is the absence of non-trivial weakly over-demanded sets is similar.
\end{proof}

\section{Auctions' general framework}\label{sec::auctions}

We define next two more set collections, \emph{excess demand} and \emph{dearth demand},
that will be used for the ascending and descending auctions framework.
As noted, if an ascending auction only increase prices of
sets in excess demand, it is guaranteed to end in the minimum Walrasian price.
If the auction increases prices of a set not in excess demand, it might
not find the minimum Walrasian equilibrium.

In the same manner, if a descending auction only decreases prices of elements
in sets of dearth demand, i.e., \emph{excess supply}, it will end in the maximum Walrasian price; if it
decreases prices of a set not in dearth demand, it might not be able to find
the maximum Walrasian price vector.

\begin{definition} [Excess demand]
Given $p\in \N^m$, $S\subseteq \Omega$ is in excess demand if $S\in OD(p)$ and $\forall T\subseteq S,\ \mbox{if } T\neq \emptyset,\mbox{ then } T\notin WUD(p + 1_S)$.
\end{definition}

\begin{definition} [Dearth demand]
Given $p\in \N^m$, $S\subseteq \Omega$ is in dearth demand if $S\in UD(p)$ and $\forall T\subseteq S,\ \mbox{if } T\neq \emptyset,\mbox{ then } T\notin WOD(p - 1_S)$.
\end{definition}
For $p\in \N^m$ we note the collection of all excess demanded sets by $ED(p)$
and the collection of all dearth demanded sets by $DD(p)$.

It is crucial to mention that known ascending and descending auctions
are special cases of `excess' or `dearth' demand auctions. In~\cite{Ben-ZwiLaNe2013} it is proved
that the auctions of Gul and Stacchetti~\cite{GulSt2000} and of
Ausubel~\cite{Ausubel2005} are identical. These auctions increase only
a set called a \emph{minimal minimizer}\footnote{Minimal minimizer is called excess demand on~\cite{GulSt2000}.} or decrease a set called a
\emph{maximal minimizer}. We show next that these are strict special cases of
the excess or dearth demand auctions.

\begin{definition} [Minimal minimizer~\cite{Ausubel2005}]
$S\subseteq \Omega$ is a \emph{minimal minimizer} for $p\in \N^m$ if
$L(p+1_S) < L(p)$ and $\forall T\subseteq \Omega, L(p+1_S)\leq L(p+1_T)$
where the inequality is strict for $T\subset S$.
\end{definition}
\begin{definition} [Maximal minimizer~\cite{Ausubel2005}]
$S\subseteq \Omega$ is a \emph{maximal minimizer} for $p\in \N^m$ if
$L(p-1_S) < L(p)$ and $\forall T\subseteq \Omega, L(p-1_S)\leq L(p-1_T)$
where the inequality is strict for $T\subset S$.
\end{definition}

\begin{lemma}\label{lemma::mm::is::ed}
For $p\in \N^m$,
\begin{itemize}
	\item If $S\subseteq \Omega$ is a minimal minimizer for $p$ then $S\in ED(p)$.
	\item If $S\subseteq \Omega$ is a maximal minimizer for $p$ then $S\in DD(p)$.
\end{itemize}
\end{lemma}
\begin{proof}
We will show it for minimal minimizer, the proof for maximal minimizer is similar.
If $S$ is a minimal minimizer for $p\in \N^m$ then $S\in OD(p)$
by definition. If $S\notin ED(p)$ then $\exists T\subset S$ such that $T\in WUD(p+1_S)$ and $T\neq \emptyset$,
then $L(p+1_S-1_T)\leq L(p+1_S)$ but then since $T\subset S$, $L(p+1_{S\setminus T})\leq L(p+1_S)$,
hence $S\setminus T$ is a witness for $S$ not being $p$'s minimal minimizer\footnote{The other direction
does not hold, see proof in the appendix, hence the new definitions strictly generalize the formers.}.
\end{proof}
Note also that, after proving the correctness of excess or dearth demand auctions framework, new ascending or descending auctions (and old ones)
can be proved easily by showing that only sets in excess demand are increased
or only sets in dearth demand are decreased. In particular, Theorem~\ref{alg::exc} and Lemma~\ref{lemma::mm::is::ed}
provide a concise proof for the well known Gul and Stacchetti's auction~\cite{GulSt2000}.

\subsection{Ascending}
We want to show next that an \emph{excess demand auction framework}, i.e., an ascending auction
that only increases sets of excess demand and terminates when there is none, will find the minimum Walrasian price vector.
By Theorem~\ref{price::char} it is enough to show that for the resulting price of the
auction, there is no over-demanded set and no non-trivial weakly under-demanded set.

\begin{lemma}\label{no::new::wud}
If for $p\in \N^m$ there is no non-trivial set in weakly under-demand, $WUD(p) = \{\emptyset\}$
and $S\subseteq \Omega$ is in excess demand, $S\in ED(p)$, then there is no non-trivial set in weakly
under-demand after increasing the price vector for every element of $S$, i.e., $WUD(p + 1_S) = \{\emptyset\}$.
\end{lemma}
\begin{proof}
Recall that by Corollary~\ref{L::char}, $\forall T,\ L(p-1_T) = L(p) + h^p(T) - |T|$.
The assumption that $WUD(p) = \{\emptyset\}$ then implies $\forall T\neq \emptyset,\ L(p-1_T) > L(p)$.
Assume that $\exists T\in WUD(p + 1_S)$ and $T\neq \emptyset$. Since $S\in ED(p)$, $T\not\subseteq S$;
by Corollary~\ref{L::char}, $L(p + 1_S - 1_T) \leq L(p + 1_S)$.

Assume first that $S\cap T = \emptyset$, then $(p + 1_S - 1_T)\wedge p = p - 1_T$ and $(p + 1_S - 1_T)\vee p = p + 1_S$.
We apply the Lyapunov submodularity, Lemma~\ref{L::SM} and get that $L(p + 1_S) + L(p - 1_T) \leq L(p) + L(p + 1_S - 1_T)$.
Now by the assumption that $L(p + 1_S - 1_T) \leq L(p + 1_S)$ we conclude that $L(p-1_T) \leq L(p)$
which is a contradiction to the assumption that there are no non-trivial weakly under-demanded sets in $p$.

Now if $S\cap T \neq \emptyset$ 
by the definition of excess demand we know that $S\cap T \notin WUD(p+1_S)$, hence, $h^{p+1_S}(S\cap T) > |S\cap T|$
and by Corollary~\ref{L::char}, $L(p+1_S)  = L(p+1_S-1_{S\cap T}) - h^{p+1_S}(S\cap T) + |S\cap T| < L(p+1_{S\setminus T})$.
By the assumption that $T\in WUD(p+1_S)$ together with Corollary~\ref{L::char} we know that
$L(p+1_{S\setminus T}-1_{T\setminus S}) = L(p+1_S-1_T) \leq L(p+1_S)$ so we have that
$L(p+1_{S\setminus T}-1_{T\setminus S}) < L(p+1_{S\setminus T})$.
But $\{S\setminus T\} \cap \{T\setminus S\} = \emptyset$ so we can again conclude that
$L(p - 1_{T\setminus S}) \leq L(p)$ which is, since $T\setminus S \neq \emptyset$, a contradiction to the
assumption that $WUD(p) = \{\emptyset\}$.
\end{proof}

\begin{lemma}\label{od::implies::ed}
$OD(p) \neq \emptyset \rightarrow ED(p) \neq \emptyset$.
\end{lemma}
\begin{proof}
For $S\in OD(p)$ either $S\in ED(p)$ or $\exists T \subseteq S$
such that $\emptyset \neq T\in WUD(p+1_S)$. But then by Corollary~\ref{L::char},
we have that $L(p+1_S-1_T) = L(p+1_S) + h^{p+1_S}(T) - |T| \leq L(p+1_S) = L(p) - l^p(S) + |S| < L(p)$
where the first inequality is the fact that $T\in WUD(p+1_S)$ and
the second comes from $S\in OD(p)$.
Hence, since $T\subseteq S$, we conclude that $S\setminus T\in OD(p)$; but since $T\neq \emptyset$
and $T\neq S$ we can proceed inductively with $S\setminus T$ instead of $S$.
\end{proof}

\begin{theorem}\label{alg::exc}
If all valuations are monotone gross substitute then any ascending auction that only updates prices of all items
that belong to an excess demand set and terminates when there is none, will result in the minimum Walrasian price vector.
\end{theorem}
\begin{proof}
Every excess ascending auction starts with a price vector $p^{\mbox{start}} = 0^m$ which guarantees
that $WUD(p^{\mbox{start}}) = \{\emptyset\}$; therefore, by Lemma~\ref{no::new::wud}, also it ends
in a price vector $p^{\mbox{end}}$, for which $WUD(p^{\mbox{end}}) = \{\emptyset\}$.
By Lemma~\ref{od::implies::ed}, if $\exists S\in OD(p^{\mbox{end}})$ then also
$\exists T\in ED(p^{\mbox{end}})$ so $p^{\mbox{end}}$ is not the result of the auction.
We conclude that $OD(p^{\mbox{end}}) = \emptyset$ and $WUD(p^{\mbox{end}}) = \{\emptyset\}$, hence,
by Theorem~\ref{price::char}, $p^{\mbox{end}}$ is the minimum Walrasian price vector.
\end{proof}

The last theorem shows the `sufficiency' of the excess demand auction framework. To show
the `necessity' we state a lemma saying that if we increase prices of a non-excess
demand set then for sure we will create a weakly under-demanded set.
Informally, note that this is `best possible' by the nature of lower bounds,
as one cannot show that other algorithms will fail on \emph{all} inputs.

\begin{lemma}\label{lb::exc}
$\emptyset \neq S\notin ED(p) \rightarrow WUD(p+1_S) \neq \{\emptyset\}$.
\end{lemma}
\begin{proof}
If $S\notin ED(p)$ then either $S\in OD(p) \wedge \exists T\in WUD(p + 1_S)$, $T\neq \emptyset$, hence we are done;
or $S\notin OD(p)$, but then, by Corollary~\ref{L::char}, since $L(p)\leq L(p+1_S)$, $S\in WUD(p + 1_S)$.
\end{proof}

\subsection{Descending}
\begin{lemma}
If for $p\in \N^m$ there is no $\emptyset \neq S\subseteq \Omega$ in weakly over-demand, $WOD(p) = \{\emptyset\}$
and $S\subseteq \Omega$ is in dearth demand, $S\in DD(p)$, then there is no non-trivial set in weakly
over-demand after decreasing the price vector for every element of $S$, i.e., $WOD(p - 1_S) = \{\emptyset\}$.
\end{lemma}
\begin{proof}
Recall that by Corollary~\ref{L::char} $\forall T,\ L(p+1_T) = L(p) - l^p(T) + |T|$.
The assumption that $WOD(p) = \{\emptyset\}$ then implies $\forall T\neq \emptyset,\ L(p+1_T) > L(p)$.
Assume that $\exists T\in WOD(p - 1_S)$ then by Corollary~\ref{L::char}, $L(p - 1_S + 1_T) \leq L(p - 1_S)$.

Assume first that $S\cap T = \emptyset$, then $(p - 1_S + 1_T)\vee p = p + 1_T$ and $(p - 1_S + 1_T)\wedge p = p - 1_S$.
We apply the submodularity lemma and get that $L(p - 1_S) + L(p + 1_T) \leq L(p) + L(p - 1_S + 1_T)$.
Now by the assumption that $L(p - 1_S + 1_T) \leq L(p - 1_S)$ we conclude that $L(p+1_T) \leq L(p)$
which is a contradiction to the assumption that there are no weakly over-demanded sets in $p$.

Now if $S\cap T \neq \emptyset$ we will show that also $T\setminus S\in WOD(p-1_S)$ which we just contradicted.
Indeed, by the definition of dearth demand we know that $S\cap T \notin WOD(p-1_S)$, hence, $l^{p-1_S}(S\cap T) < |S\cap T|$
and by Corollary~\ref{L::char}, $L(p-1_S) < L(p-1_{S\setminus T})$.
By the assumption that $T\in WOD(p-1_S)$ together with Corollary~\ref{L::char} we know that
$L(p-1_{S\setminus T}+1_{T\setminus S}) = L(p-1_S+1_T) \leq L(p-1_S)$ so we have that
$L(p-1_{S\setminus T}+1_{T\setminus S}) \leq L(p-1_{S\setminus T})$. Now clearly
$(p - 1_{S\setminus T} + 1_{T\setminus S})\vee p = p + 1_{T\setminus S}$ and
$(p - 1_{S\setminus T} + 1_{T\setminus S})\wedge p = p - 1_{S\setminus T}$;
by Lyapunov submodularity then
$L(p - 1_{S\setminus T}) + L(p + 1_{T\setminus S}) \leq L(p) + L(p - 1_{S\setminus T} + 1_{T\setminus S})$.
By the fact that $L(p-1_{S\setminus T}+1_{T\setminus S}) \leq L(p-1_{S\setminus T})$
we conclude that $L(p + 1_{T\setminus S}) \leq L(p)$ which is a contradiction to the
assumption that $WOD(p) = \{\emptyset\}$.
\end{proof}

\begin{lemma}
$UD(p) \neq \emptyset \rightarrow DD(p) \neq \emptyset$.
\end{lemma}
\begin{proof}
For $S\in UD(p)$ either $S\in DD(p)$ or $\exists T \subseteq S$
such that $\emptyset \neq T\in WOD(p-1_S)$. But then by Corollary~\ref{L::char}
we have that $L(p-1_S+1_T) = L(p-1_S) - l^{p-1_S}(T) + |T| \leq L(p-1_S) = L(p) + h^p(S) - |S| < L(p)$
where the first inequality is the fact that $T\in WOD(p-1_S)$ and
the second comes from $S\in UD(p)$.
Hence we conclude that $S\setminus T\in UD(p)$; but since $T\neq \emptyset$
and $T\neq S$ we can proceed inductively with $S\setminus T$ instead of $S$.
\end{proof}

\begin{theorem}\label{alg::dea}
If all valuations are monotone gross substitute then any descending auction that only updates prices of
sets in dearth demand and terminates only when there is no dearth demand set, will find
the maximum Walrasian price vector.
\end{theorem}

\begin{lemma}\label{lb::dea}
$\emptyset \neq S\notin DD(p) \rightarrow WOD(p-1_S) \neq \{\emptyset\}$.
\end{lemma}
\begin{proof}
If $S\notin DD(p)$ then either $S\in UD(p) \wedge \exists T\in WOD(p - 1_S)$, $T\neq \emptyset$, hence we are done;
or $S\notin UD(p)$, but then, by Corollary~\ref{L::char}, since $L(p)\leq L(p-1_S)$, $S\in WOD(p - 1_S)$.
\end{proof}

\section{Conclusion}
We introduced two new definitions for monotone gross substitute valuations;
in Paes~Leme's recent survey one can find many more interesting definitions and elegant
characterizations~\cite{PaesLeme2014}.

As noted, our new price characterization generalizes Gul and Stacchetti's from \emph{envy free} price
if and only if there is no \emph{over demanded} set, to \emph{Walrasian} price vector
if and only if there is no \emph{over demanded} nor \emph{under demanded} set.
We showed that \emph{minimum} and \emph{maximum} Walrasian prices can also be characterized using
\emph{weakly} over and under demand definition.

We introduced new frameworks that all ascending and descending auctions must follow
in order to guarantee termination at a minimum or maximum Walrasian. We showed, clearly, that
former auctions follow that framework.

\section{Acknowledgment}
I would like to thank Eyal Ackerman, Monika Henzinger, Oded Lachish,
Veronika Loitzenbauer, and Guy Wolfovitz for
the help in improving the presentation of the results.

\bibliographystyle{plain}
\bibliography{biblio}

\appendix
\section{More known results} \label{app::more}
Blumrosen and Nisan showed that if a set of gross substitute valuations is not additive when \emph{aggregated}
then there exists another valuation such that together there is no ascending auction that finds the maximum social
welfare~\cite{BlumrosenNi2010}. Gul and Stacchetti showed that no ascending auction can find a \textbf{VCG} outcome
when all the valuations are gross substitute~\cite{GulSt2000}.

Gul and Stacchetti showed also that when valuations are monotone the class monotone gross substitute is
equivalent to the class of monotone \emph{single improvement} valuations which is defined by the following.
\begin{definition} [Single improvement~\cite{GulSt1999}]
A valuation is in the class \emph{single improvement} if for every $p\in \N^m$ and a set
of items $S\notin D(p)$, there exists $T\subseteq \Omega$, such that $u_p(S) < u_p(T)$, $|T\setminus S|\leq 1$ and
$|S\setminus T|\leq 1$.
\end{definition}

For non-gross substitute valuations,
the following theorem was proved by Reijnierse~et~al.~\cite{ReijnierseGePo2002}.
\begin{theorem} [non-gross substitute~\cite{ReijnierseGePo2002}]\label{non::gs}
For a non-gross substitute valuation of more than two items,
$\exists q\in \N^m$, $A\subseteq \Omega$, and three items $j_1,j_2,j_3\in \Omega\setminus A$,
such that either $D(q) = \{A,A\cup\{j_1,j_2\}\}$ or $D(q) = \{A\cup\{j_3\},A\cup\{j_1,j_2\}\}$.\footnote{For valuations
of two items it is known that gross substitute is equivalent to submodularity, hence, the first case holds for
non-gross substitute valuations of two items. Therefore, the restriction of the number of items is only for the sake
of representation.}
\end{theorem}
When all valuations are gross substitute it can be shown that the maximum utility function
is submodular with respect to the price vector~\cite{GulSt2000}. Hence, since the Lyapunov is a summation
of the maximum utility functions and linear in the price, the following \emph{Lyapunov submodularity}
can be proved. This result appears for instance in Ausubel~\cite{Ausubel2005};
for a direct and complete proof, see~\cite{Ben-ZwiLaNe2013}.
\begin{lemma} [Lyapunov submodularity~\cite{Ausubel2005}]\label{L::SM}
If all valuations are gross substitute then for any two $p,q\in \N^m$,
$L(p\wedge q) + L(p\vee q) \leq L(p) + L(q)$.
\end{lemma}
The notions of \emph{maximum Walrasian price vector}
and \emph{minimum Walrasian price vector} are well defined
and represent unique price vectors by a direct corollary to the following lemma.
\begin{lemma} [Walrasian lattice property~\cite{GulSt2000}]\label{lattice}
If $p,q\in \N^m$ are two Walrasian price vectors then so is $(p\wedge q)$
and $(p\vee q)$.
\end{lemma}
The following result is used in the proof of Theorem~\ref{price::char}.
\begin{theorem} [Theorem 2 in~\cite{GulSt2000}]\label{thm::2::GS}
For monotone gross substitute valuation and two sets of items $S,T$ such that $S\cap T = \emptyset$
it holds that $l_p(S) \geq l_{p+1_T}(S)$.
\end{theorem}

Ascending auctions for generalized price definitions were suggested, for instance
by~\cite{MishraPa2007} and~\cite{VriesScVo2007}. The main difference In the model is allowing
non anonymous nonlinear prices. That is, the price of each set for each bidder can
be arbitrary. Our notion of 'over-demand' is analog to de~Vries~et.~al's notion of
'undersupplied' and their notion of 'minimally undersupplied' is analog to a subset
of the 'excess-demand' we define. They propose an ascending auction, for these relaxed
prices, that computes VCG prices for submodular valuations. Mishra and Parks suggest
an ascending auction that uses Walrasian equilibrium in order to find a VCG price for
all valuations. Apart from the price relaxations, the notions of 'ascending' and of
'Walrasian equilibrium' are also relaxed compare to our definitions.

\section{Excess demand set system strictly generalizes (previous definition of) excess demand set}
We finish with a lemma which proves that the newly defined set system \emph{excess demand} strictly generalizes
the known set definition \emph{minimal minimizer} (also known as excess demand in~\cite{GulSt2000}).
\begin{lemma}
$\exists p,S,\ s.t.\ S\in ED(p)\wedge S$ is not a minimal minimizer for $p$.
\end{lemma}
\begin{proof}
Assume there are three identical additive bidders and two items they all
evaluate for $1$. For $p = (0,0)$ each singleton is in $ED(p)$ but the minimal
minimizer is the set of both items.
\end{proof}

\end{document}